\documentclass[11pt,letterpaper]{article}

\usepackage{graphicx}
\usepackage{url}
\usepackage{amsfonts}
\usepackage{amsmath}
\usepackage{amsthm}
\usepackage{amssymb}
\usepackage{url}
\usepackage{subfig}
\usepackage[colorlinks=true,citecolor=black,linkcolor=black,urlcolor=blue]{hyperref}
\usepackage{fullpage}
\usepackage{xcolor}
\usepackage[mathlines]{lineno}

\newcommand*\patchAmsMathEnvironmentForLineno[1]{%
  \expandafter\let\csname old#1\expandafter\endcsname\csname #1\endcsname
  \expandafter\let\csname oldend#1\expandafter\endcsname\csname end#1\endcsname
  \renewenvironment{#1}%
     {\linenomath\csname old#1\endcsname}%
     {\csname oldend#1\endcsname\endlinenomath}}%
\newcommand*\patchBothAmsMathEnvironmentsForLineno[1]{%
  \patchAmsMathEnvironmentForLineno{#1}%
  \patchAmsMathEnvironmentForLineno{#1*}}%
\AtBeginDocument{%
\patchBothAmsMathEnvironmentsForLineno{equation}%
\patchBothAmsMathEnvironmentsForLineno{align}%
\patchBothAmsMathEnvironmentsForLineno{flalign}%
\patchBothAmsMathEnvironmentsForLineno{alignat}%
\patchBothAmsMathEnvironmentsForLineno{gather}%
\patchBothAmsMathEnvironmentsForLineno{multline}%
}

\allowdisplaybreaks

\newtheorem{theorem}{Theorem}[section]
\newtheorem{proposition}[theorem]{Proposition}
\newtheorem{lemma}[theorem]{Lemma}

\newcommand{\M}{\ensuremath{\mathcal{M}}}
\newcommand{\E}{\ensuremath{\mathcal{E}}}
\newcommand{\D}{\ensuremath{\mathcal{B}}}

\title{On maximum-sum matchings of bichromatic points}

\author{
	Oscar Chacón-Rivera\thanks{Universidad de Santiago de Chile (USACH), Facultad de Ciencia, Departamento de Matem\'atica y Ciencia de la Computaci\'on, Chile {\tt \{oscar.chacon,pablo.perez.l\}@usach.cl}.}
	\and
	Pablo P\'erez-Lantero\footnotemark[1]
}

\begin{document}

\maketitle

\begin{abstract}
Huemer et al.\ (Discrete Math, 2019) proved that for any two finite point sets $R$ and $B$ in the plane with $|R| = |B|$, the perfect matching that matches points of $R$ with points of $B$, and maximizes the total {\em squared} Euclidean distance of the matched pairs, has the property that all the disks induced by the matching have a nonempty common intersection.
A pair of matched points induces the disk that has the segment connecting the points as diameter.
In this note, we characterize these maximum-sum matchings for some family of continuous (semi-)metrics, focusing on both the Euclidean distance and squared Euclidean distance. Using this characterization, we give a different but simpler proof for the common intersection property proved by Huemer et al..
\end{abstract}

\section{Introduction}
Let $R$ and $B$ be two point sets in the plane with $|R| = |B|$. The points in $R$ are called \textit{red}, and those in $B$ are called \textit{blue}. A \textit{matching} $\M$ of $R \cup B$ is a partition of $R \cup B$ into $n$ pairs such that each pair consists of a red point and a blue point. A point $p \in R$ and a point $q \in B$ are matched if and only if the (unordered) pair $(p,q)$ is in the matching. For every $p,q \in \mathbb{R}^2$, we use $pq$ to denote the segment connecting $p$ and $q$, $\ell(pq)$ to denote the line through $pq$, and $\D(pq)$ to denote the disk with diameter equal to the length $\|p-q\|$ of $pq$, that is centered at the midpoint $(p + q)/2$ of $pq$. For any matching $\M$, we use $\D_{\M}$ to denote the set of the disks associated with, or induced by, the matching, that is, $\D_{\M} = \{ \D(pq):(p,q) \in \M \}$, where $\D(pq)$ is the disk induced by the edge $(p,q)$.

Given a metric, or semi-metric function $d:\mathbb{R}^2\times \mathbb{R}^2\rightarrow \mathbb{R}_{\ge 0}$, we say that a matching $\M$ is {\em max-sum} if it maximizes $\sum_{(p,q)\in\M} d(p,q)$ among all matchings of $R$ and $B$. Recall that a function is a semi-metric if it satisfies all the properties of a metric function except the triangle inequality. 

Huemer et al.~\cite{HUEMER2019} proved that if $\M$ is a max-sum matching for $d(p,q)=\| p - q \|^2$ for all $p,q\in \mathbb{R}^2$, then all the disks of $\D_{\M}$ have a point in common. That is, $d$ is the squared Euclidean distance, also called {\em quadrance}, and it is in fact a semi-metric.

For the Euclidean distance $d(p,q)=\|p-q\|$, Bereg et al.~\cite{bereg2022} proved that any max-sum matching $\M$ between $R$ and $B$ does not always satisfy that all disks in $\D_M$ have a common point, although the disks intersect pairwise. As their main result, they proved that the common intersection is always satisfied when $\M$ is a max-sum matching between $2n$ uncolored points. For an even point set in the plane, a matching is simply a partition of the set into pairs.

In this note, we consider max-sum matchings between planar colored point sets $R$ and $B$ with $|R| = |B|=n$, $n\ge 1$, for {\em any continuous} (semi-)metric $d$ of the Euclidean plane satisfying the following properties:
\begin{itemize}
    \item[P1.] $d(x,y) \ge 0$ for all $x,y \in \mathbb{R}^2$.
    \item[P2.] $d(x,x) = 0$ for all $x \in \mathbb{R}^2$.
    \item[P3.] $d(x,y) = d(y,x)$ for all $x,y \in \mathbb{R}^2$.
    \item[P4.] For any fixed points $p,q \in \mathbb{R}^2$ with $p \neq q$, and any real constant $t$, the level set $\{ z \in \mathbb{R}^2 : d(p,z) - d(q,z) = t \}$ is connected.
\end{itemize} Note that the first three properties are the standard properties of any semi-metric, whereas the fourth property is satisfied by semi-metrics such as the Euclidean distance and the squared Euclidean distance.

In Section~\ref{sec:characterization}, we characterize the max-sum matchings between {\em three} red points and {\em three} blue points in terms of the common intersection of certain sets defined by the six points. Using this characterization, in the case where $d$ is the Euclidean quadrance, we give in Section~\ref{sec:proof} an elementary proof to the main theorem by Huemer et al.~\cite{HUEMER2019}: The disks induced by a max-sum matching have a common point. 

\subsection{Related work}

Motivated by the study of geometric matchings of planar point sets, Huemer et al.~\cite{HUEMER2019} proved that for any finite bichromatic point set $R \cup B$ with $|R|=|B|$, and a max-sum matching $\M$ of $R$ and $B$ according to the squared Euclidean distance, all disks in $\D_{\M}$ have a nonempty common intersection. 
Later, Bereg et al.~\cite{bereg2022} proved that all disks in $\D_{\M}$ have a nonempty common intersection for any point set $P$ of $2n$ uncolored points in the plane, when $\M$ is a max-sum matching of $P$ with respect to the Euclidean distance.
This result has been slightly strengthened by Barabanshchikova and Polyanskii~\cite{barabanshchikova2024}, who proved that the {\em interiors} of all disks in $\D_{\M}$ have a nonempty common intersection in the case where all elements of $P$ are distinct. 

Sober\'on and Tang~\cite{soberon2020} considered this kind of geometric intersection problems in the more general context of geometric graphs: Given a geometric graph with vertex set a finite point set in the plane (where the edges are straight segments connecting points), they define the graph as a {\em Tverberg graph} if the disks induced by the edges of the graph have a nonempty common intersection. They show that for any odd planar point set there exists a Hamiltonian cycle which is a Tverberg graph, and that for any even planar point set there exists a Hamiltonian path with the same property. 
Notice that the previous mentioned results can be stated in terms of Tverberg graphs. For example, the result of Huemer et al.~\cite{HUEMER2019}: For any finite bichromatic point set $R \cup B$ with $|R|=|B|$, any max-sum matching of $R$ and $B$ according to the squared Euclidean distance is a Tverberg graph.


Pirahmad et al.~\cite{pirahmad2022} refined and extended these results. They proved, for example, that: for any finite point set in the plane there exists a Hamiltonian cycle that is a Tverberg graph; for any even point set in $\mathbb{R}^d$ there exists a matching that is a Tverberg graph; and for any red and blue points in $\mathbb{R}^d$, $d\ge 3$, there exists a perfect red-blue matching that is a Tverberg graph. This last result generalizes the initial result of Huemer et al.~\cite{HUEMER2019} to higher dimensions.

In the same direction, Abu-Affash et al.~\cite{abu2023piercing} proved that for any finite planar point set $P$, the maximum-weight spanning tree of $P$ is a Tverberg graph. In fact, they proved that the center of the smallest enclosing circle of $P$ is contained in all the disks induced by the tree. 

Fingerhut (see Eppstein~\cite{andy}), motivated by a problem in designing communication networks (see Fingerhut et al.~\cite{FingerhutST97}), conjectured that given a set $P$ of $2n$ uncolored points in the plane, not necessarily distinct, and a max-sum matching $\{(a_i,b_i),i=1,\dots,n\}$ of $P$, there exists a point $o$ of the plane, not required to be a point of $P$ and called the center of the matching, such that
\begin{equation}\label{eqFingerhut}
    \|a_i-o\|+\|b_i-o\|~\le~ \frac{2}{\sqrt{3}}~\|a_i-b_i\| \hspace{0.3cm} \textrm{for all } i\in\{1,\ldots,n\}, ~\textrm{where } 2/\sqrt{3}\approx 1.1547.
\end{equation}
Bereg et al.~\cite{bereg2022}, by proving that for any point set $P$ of $2n$ uncolored points in the plane and a max-sum matching $\M=\{(a_i,b_i),i=1,\dots,n\}$ of $P$ we have that all disks in $\D_{\M}$ have a nonempty common intersection, obtained an approximation to this conjecture. Indeed, any point $o$ in the common intersection satisfies
\[
	\|a_i-o\|+\|b_i-o\|~\le~ \sqrt{2}~\|a_i-b_i\|, ~\textrm{where } \sqrt{2}\approx 1.4142.
\]
Recently, Barabanshchikova and Polyanskii~\cite{barabanshchikova2022} confirmed the conjecture of Fingerhut. After that, Pérez-Lantero and Seara~\cite{perez2024center} proved the bichromatic version of the main result obtained by Bereg et al.~\cite{bereg2022}: If $R$ and $B$ are two point sets in the plane with $|R|=|B|=n$, and $\M=\{(r_i,b_i)\in R\times B:i=1,2,\ldots,n\}$ is a max-sum matching for the Euclidean distance, there exists a point $o$ of the plane such that
\begin{equation}\label{eq2}
	\|r_i-o\|+\|b_i-o\|\le \sqrt{2}~\|r_i-b_i\| \hspace{0.3cm} \textrm{for all } i\in\{1,2,\ldots,n\}.
\end{equation}
These results about the conjecture of Fingerhut are related to the common intersection of convex sets induced by the pairs of the max-sum matchings, ellipses in this case. Indeed,
for two points $p,q\in\mathbb{R}^2$ and a real number $\lambda\ge 1$, let $\E_\lambda(pq)$ denote the (region bounded by) the ellipse with foci $p$ and $q$ and major axis length $\lambda\|p-q\|$. That is, $\E_\lambda(pq)=\{x\in\mathbb{R}^2:\|p-x\|+\|q-x\|\le \lambda\|p-q\|\}$.  The statement of Equation~\eqref{eqFingerhut} is equivalent to stating that the common intersection of the ellipses $\E_{2/\sqrt{3}}(a_1b_1), \E_{2/\sqrt{3}}(a_2b_2),\dots, \E_{2/\sqrt{3}}(a_nb_n)$ is nonempty. Furthermore, the statement of Equation~\eqref{eq2} is equivalent to stating that the common intersection of the ellipses $\E_{\sqrt{2}}(r_1b_1)$, $\E_{\sqrt{2}}(r_2b_2),\dots, \E_{\sqrt{2}}(r_nb_n)$ is also nonempty~\cite{barabanshchikova2022,andy,perez2024center}.

\section{Characterization of max-sum matchings of 6 points}
\label{sec:characterization}

Let $R=\{a,b,c\}$ be a set of three red points and $B = \{ a', b', c' \}$ a set of three blue points. Let $d$ be a continuous (semi-)metric function on $\mathbb{R}^2$ satisfying properties P1, P2, P3, and P4. Let $\{ (a,a'), (b,b'), (c,c') \}$ be a matching of $R$ and $B$.
Let us define the following six sets:
\begin{align*}
 H(a,b) = & \{ x \in \mathbb{R}^2 : d(a,b') - d(b,b') \le d(a,x) - d(b,x) \}, \\
 H(b,c) = & \{ x \in \mathbb{R}^2 : d(b,c') - d(c,c') \le d(b,x) - d(c,x) \}, \\
 H(c,a) = & \{ x \in \mathbb{R}^2 : d(c,a') - d(a,a') \le d(c,x) - d(a,x) \},
\end{align*} and 
\begin{align*}
 h(a,b) = & \{ x \in \mathbb{R}^2 : d(a,x) - d(b,x) \le d(a,a') - d(b,a') \}, \\
 h(b,c) = & \{ x \in \mathbb{R}^2 : d(b,x) - d(c,x) \le d(b,b') - d(c,b') \}, \\
 h(c,a) = & \{ x \in \mathbb{R}^2 : d(c,x) - d(a,x) \le d(c,c') - d(a,c') \}.
\end{align*}
In other words, $H(a,b)$ is the set of the (blue) points $x\in\mathbb{R}^2$ such that $d(a,b')+d(b,x)\le d(a,x)+d(b,b')$. That is, $\{(a,x),(b,b')\}$ is a max-sum matching of $\{a,b\}$ and $\{x,b'\}$. Similarly, $h(a,b)$ is the set of the (blue) points $x\in\mathbb{R}^2$ such that $\{(a,a'),(b,x)\}$ is a max-sum matching of $\{a,b\}$ and $\{a',x\}$.
%
%
For example, if $d$ is the Euclidean distance, then the boundaries of the sets $H(a,b)$ and $h(a,b)$ are arcs of hyperbolas with foci $a$ and $b$ (see Figure~\ref{fig:H-sets-Euclidean}). Similarly, if $d$ is the {\em squared} Euclidean distance, then $H(a,b)$ and $h(a,b)$ are half-planes whose boundaries are perpendicular to the line $\ell(ab)$ through the segment $ab$ (see Figure~\ref{fig:H-sets-sqEuclidean}).

\begin{figure}[t]
 \centering
 \subfloat[]{
  \includegraphics[page=1,scale=0.8]{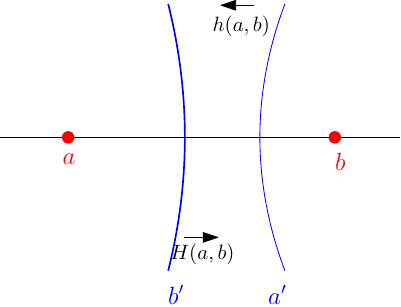}
  \label{fig:H-sets-Euclidean}
 }~~~~
 \subfloat[]{
  \includegraphics[page=2,scale=0.8]{img.pdf}
  \label{fig:H-sets-sqEuclidean}
 }
 \caption{\small{Definition of the sets $H(a,b)$ and $h(a,b)$ for (a) $\| \cdot \|$, and (b) $\| \cdot \|^2$ (the labels $b'$ and $a'$ are used to indicate that they are in the boundaries of $H(a,b)$ and $h(a,b)$, respectively).}}
 \label{fig:H-h-sets}
\end{figure}

\begin{lemma}\label{lem:1}
If the five intersections $H(a,b) \cap H(b,c) \cap H(c,a)$, $h(a,b) \cap h(b,c) \cap h(c,a)$, $H(a,b)\cap h(a,b)$, $H(b,c)\cap h(b,c)$, and $H(c,a)\cap h(c,a)$ are all nonempty, then $\{ (a,a'),(b,b'),(c,c') \}$ is a max-sum matching of $\{a,b,c\}$ and $\{a',b',c'\}$.
\end{lemma}

\begin{proof}
Let $z_1 \in H(a,b) \cap H(b,c) \cap H(c,a)$. Then,
\begin{align*}
  d(a,b') - d(b,b') & \le d(a,z_1) - d(b,z_1), \\
  d(b,c') - d(c,c') & \le d(b,z_1) - d(c,z_1), \\
  d(c,a') - d(a,a') & \le d(c,z_1) - d(a,z_1).
\end{align*} 
By adding the above three equations, we obtain that
\[
	d(a,b') + d(b,c') + d(c,a') \le d(a,a') + d(b,b') + d(c,c').
\] 
Let $z_2 \in h(a,b) \cap h(b,c) \cap h(c,a)$. Then, 
 \begin{align*}
  d(a,z_2) - d(b,z_2) & \le d(a,a') - d(b,a'), \\
  d(b,z_2) - d(c,z_2) & \le d(b,b') - d(c,b'), \\
  d(c,z_2) - d(a,z_2) & \le d(c,c') - d(a,c').
 \end{align*} 
By adding the above three equations, we obtain that 
\[ 
	d(a,c') + d(b,a') + d(c,b') \le d(a,a') + d(b,b') + d(c,c'). 
\]
Let $z_3\in H(a,b)\cap h(a,b)$. By the definitions of $H(a,b)$ and $h(a,b)$, we have that
\[
	d(a,b') - d(b,b') \le d(a,z_3) - d(b,z_3) \le d(a,a') - d(b,a').
\]
Then, we obtain $d(a,b') + d(b,a') \le d(a,a') + d(b,b')$, which implies that
\[
	d(a,b') + d(b,a') + d(c,c') \le d(a,a') + d(b,b') + d(c,c').
\]
Analogously, we obtain that
\begin{align*}
	d(a,a') + d(b,c') + d(c,b') & \le d(a,a') + d(b,b') + d(c,c'), ~~\text{and}\\
	d(a,c') + d(b,b') + d(c,a') & \le d(a,a') + d(b,b') + d(c,c').
\end{align*}
Therefore, $\{(a,a'),(b,b'),(c,c')\}$ is a max-sum matching of $\{a,b,c\}$ and $\{a',b',c'\}$, since its total sum is never smaller than that of any other matching.
\end{proof}

\begin{lemma}\label{lem:2}
If $\{ (a,a'), (b,b'), (c,c') \}$ is a max-sum matching of $\{a,b,c\}$ and $\{a',b',c'\}$, then the five intersections $H(a,b) \cap H(b,c) \cap H(c,a)$, $h(a,b) \cap h(b,c) \cap h(c,a)$, $H(a,b)\cap h(a,b)$, $H(b,c)\cap h(b,c)$, and $H(c,a)\cap h(c,a)$ are all nonempty.
\end{lemma}

\begin{proof}
Observe that $\{ (a,a'), (b,b') \}$ must be a max-sum matching of $\{ a,b\}$ and $\{a',b' \}$. Then, we have that $d(a,b')+d(b,a')\le d(a,a')+d(b,b')$, which is equivalent to $d(a,b')-d(b,b')\le d(a,a')-d(b,a')$. This implies that $a',b' \in H(a,b) \cap h(a,b)$, which means that $H(a,b) \cap h(a,b)$ is not empty. Analogously, $H(b,c) \cap h(b,c)$ and $H(c,a) \cap h(c,a)$ are not empty.

Let us prove now that $H(a,b) \cap H(b,c) \cap H(c,a) \neq \emptyset$. The proof for $h(a,b) \cap h(b,c) \cap h(c,a)$ is analogous. We divide the proof into two cases.
 
As the first case, suppose that $a' \in H(b,c)$, $b' \in H(c,a)$, or $c' \in H(a,b)$. Assume w.l.o.g.\ that $a' \in H(b,c)$.

By the definition and properties of the sets $H(\cdot, \cdot)$, we have that $a' \in H(a,b) \cap H(c,a)$, which implies that $a' \in H(a,b) \cap H(b,c) \cap H(c,a)$, and the result thus follows.

Let $U,V \subset \mathbb{R}^2$ be two closed sets with nonempty connected boundaries. It is a known result from point-set topology that if the three sets $U \setminus V$, $U \cap V$, and $\partial V \setminus U$ are all nonempty, then there exists a point $z \in \partial U \cap \partial V$. 

That said, suppose now that $a' \notin H(b,c)$, $b' \notin H(c,a)$, and $c' \notin H(a,b)$, as the second case of the proof. We then have that $a' \in H(a,b) \setminus H(b,c)$, $b' \in H(a,b) \cap H(b,c)$, and $c' \in \partial H(b,c) \setminus H(a,b)$. Given that $H(a,b)$ and $H(b,c)$ are closed regions of the plane, due to the fact that $d$ is continuous, and also that $\partial H(a,b)$ and $\partial H(b,c)$ are connected and nonempty by the assumed properties of the (semi-)metric $d$, the topological result above implies that there exists a point $z \in H(a,b) \cap H(b,c)$ in the boundary of $H(a,b)$ and also in the boundary of $H(b,c)$. Assume by contradiction that $z \notin H(c,a)$. Then, we have that 
\begin{align*}
  d(c,z) - d(a,z) & < d(c,a') - d(a,a'), \qquad (z \notin H(c,a)) \\
  d(a,z) - d(b,z) & = d(a,b') - d(b,b'), \qquad (z \in \partial H(a,b)) \\
  d(b,z) - d(c,z) & = d(b,c') - d(c,c'). \qquad (z \in \partial H(b,c))
\end{align*}
By adding the above three equations, we obtain that $d(a,a') + d(b,b') + d(c,c') < d(a,b') + d(b,c') + d(c,a')$, which means that $\{ (a,a'), (b,b'), (c,c') \}$ is not a max-sum matching since the matching $\{ (a, b'), (b, c'), (c, a') \}$ has larger sum, which is a contradiction. Hence, we must have that $z \in H(c,a)$, which implies $z \in H(a,b) \cap H(b,c) \cap H(c,a)$. The result thus follows.
\end{proof}

By combining Lemma~\ref{lem:1} and Lemma~\ref{lem:2}, we obtain the main result of this section:

\begin{theorem}\label{theo:characterization}
The matching $\{ (a,a'), (b,b'), (c,c') \}$ is a max-sum matching of the red points $\{a,b,c\}$ and the blue points $\{a',b',c'\}$ if and only if
the five intersections $H(a,b) \cap H(b,c) \cap H(c,a)$, $h(a,b) \cap h(b,c) \cap h(c,a)$, $H(a,b)\cap h(a,b)$, $H(b,c)\cap h(b,c)$, and $H(c,a)\cap h(c,a)$ are all nonempty.
\end{theorem}

\section{Common intersection for the Euclidean quadrance}
\label{sec:proof}

Let $R$ and $B$ be two point sets in the plane, with $|R|=|B|=n\ge 2$, and let $\M$ be a max-sum matching of $R$ and $B$ with respect to the Euclidean quadrance. Huemer et al.~\cite{HUEMER2019} first proved that the disks in $\D_\M$ intersect pairwise, so the common intersection of the disks in $\D_\M$ holds for $n=2$. To give our new proof of this property for $n\ge 3$, we will use Helly's Theorem~\cite{helly1923} in $\mathbb{R}^2$, that is, we will prove the common intersection for $n=3$, by using the characterization of Section~\ref{sec:characterization}. For completeness, we state Helly's Theorem:

\begin{theorem}[Helly~\cite{helly1923}]\label{thm:helly} 
Let $\mathcal{F}$ be a finite family of closed convex sets in $\mathbb{R}^m$ such that every subfamily of $m+1$ sets of $\mathcal{F}$ has nonempty intersection. Then, all sets in $\mathcal{F}$ have nonempty intersection. 
\end{theorem}

Let $\M=\{(a,a'),(b,b'),(c,c')\}$ be a max-sum matching of $R=\{a,b,c\}$ and $B=\{a',b',c'\}$, and let us define the three sets $S(a,b)=H(a,b)\cap h(a,b)$, $S(b,c)=H(b,c)\cap h(b,c)$, and $S(c,a)=H(c,a)\cap h(c,a)$. Given that $S(a,b)$, $S(b,c)$, and $S(c,a)$ are nonempty strips perpendicular to the lines  $\ell(ab)$, $\ell(bc)$, and $\ell(ca)$, respectively, Theorem~\ref{theo:characterization} asserts that $S(a,b) \cap S(b,c) \cap S(c,a) \neq \emptyset$ (see Figure~\ref{fig:strips}). So, let $z$ be a point of $S(a,b) \cap S(b,c) \cap S(c,a)$.

For two distinct points $p,q\in\mathbb{R}^2$, let $\tau(pq)$ be the ray with apex $p$ that goes through $q$. For three distinct points $p,q,r\in\mathbb{R}^2$, let $\omega(pqr)$ denote the convex wedge bounded by $\tau(pq)$ and $\tau(pr)$ and $\Delta pqr$ denote the triangle with vertices $p$, $q$, and $r$.

\begin{proposition}\label{prop:1}
Let $z^*$ be the orthogonal projection of the point $z$ into the line $\ell(ab)$. The following statements are satisfied:
(i) If $z^*\in ab$, then $z^*\in\D(aa')\cap \D(bb')$.
(ii) If $z^*\in \tau(ab)\setminus ab$, then $z^*, b\in \D(aa')$.
(iii) If $z^*\in \tau(ba)\setminus ab$, then $z^*, a\in \D(bb')$.
\end{proposition}

\begin{proof}
Let $a^*$ and $b^*$ be the orthogonal projections of $a'$ and $b'$ into $\ell(ab)$, respectively. Since $z,z^*\in H(a,b)\cap h(a,b)$, $a'\in \partial h(a,b)$, and  $b'\in \partial H(a,b)$, we have that $z^*$ belongs to the segment $b^*a^*$. Note that if $z^*\in aa^*$, then $z^*\in\D(aa')$ by Thales' Theorem. Similarly, if $z^*\in bb^*$, then $z^*\in\D(bb')$. Hence: 
(i) If $z^*\in ab$, then $z^*\in aa^*\cap bb^*$, which implies $z^*\in\D(aa')\cap \D(bb')$ (see Figure~\ref{fig:prop1}).
(ii) If $z^*\in \tau(ab)\setminus ab$, then $z^*, b\in aa^*$, which implies $z^*, b\in \D(aa')$ (see Figure~\ref{fig:prop2}).
(iii) If $z^*\in \tau(ba)\setminus ab$, then $z^*, a\in bb^*$, which implies $z^*, a\in \D(bb')$.
\end{proof}

\begin{figure}[t]
 \centering
 \subfloat[]{\label{fig:strips}\includegraphics[page=9,scale=1]{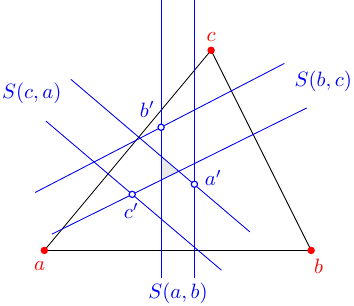}}~~~~~
 \subfloat[]{\label{fig:prop1}\includegraphics[page=10,scale=1]{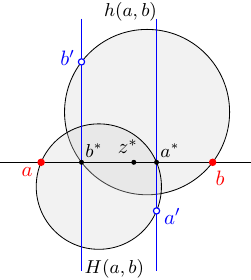}}~~~~~
 \subfloat[]{\label{fig:prop2}\includegraphics[page=11,scale=1]{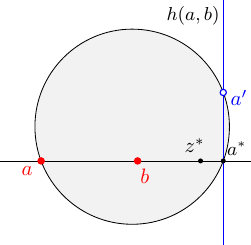}}
 \caption{\small{
 		(a) Strips $S(a,b)$, $S(b,c)$, and $S(c,a)$ with their common intersection.
		(b-c) Proof of Proposition~\ref{prop:1}.
	}}
 \label{fig:prop}
\end{figure}

We prove first the extreme case in which $a$, $b$, and $c$ are collinear points. Assume w.l.o.g.\ that $b$ belongs to the segment $ac$. Since $S(a,b)\cap S(b,c)\cap S(c,a)\neq \emptyset$, this intersection must be a line through $b'$ perpendicular to $\ell(ac)$. So, let $z^*$ be the intersection of this line with $\ell(ac)$, and note that $z^*\in\D(bb')$.
If $z^*\in ac$, then $z^*\in \D(aa')\cap \D(cc')$ by Proposition~\ref{prop:1}~(i). This implies that $z^*\in \D(aa')\cap \D(bb') \cap\D(cc')$.
If $z^*\in \tau(ac)\setminus ac$, then $c\in\D(bb')$ since $c\in bz^*$, and $c\in\D(aa')$ by Proposition~\ref{prop:1}~(ii). Hence, $c\in \D(aa')\cap \D(bb') \cap\D(cc')$.
Similarly, if $z^*\in \tau(ca)\setminus ac$, then $a\in\D(bb')$ since $a\in bz^*$, and $a\in\D(cc')$ by Proposition~\ref{prop:1}~(iii). Hence, $a\in \D(aa')\cap \D(bb') \cap\D(cc')$.
In all the cases, we have that $\D(aa')\cap \D(bb') \cap\D(cc')\neq\emptyset$.

From this point forward, assume that $a$, $b$, and $c$ are not collinear points.

\begin{proposition}\label{prop:2}
If the point $z$ belongs to the wedge $\omega(abc)$ (resp.\ $\omega(bca)$, $\omega(cab)$), then $z$ is contained in $\D(aa')$ (resp.\ $\D(bb')$, $\D(cc')$).
\end{proposition}

\begin{proof}
We prove the case $z\in\omega(abc)$. The proofs for the other cases are analogous.
Given that $z\in h(a,b)\cap H(c,a)$ and that $z\in\omega(abc)$, the lines $\partial h(a,b)$ and $\partial H(c,a)$ cut perpendicularly the rays $\tau(ab)$ and $\tau(ac)$, respectively. Let $u=\tau(ab)\cap \partial h(a,b)$ and $v=\tau(ac)\cap \partial H(c,a)$, which are the orthogonal projections of $a'$ into $\tau(ab)$ and $\tau(bc)$, respectively (see Figure~\ref{fig:prop3}). Furthermore, $z$ is in one of the right triangles $\Delta aa'u$ and $\Delta aa'v$, both contained in $\D(aa')$. Hence, $z\in\D(aa')$.
\end{proof} 

\begin{figure}[t]
 \centering
 \subfloat[]{\label{fig:prop3}\includegraphics[page=12,scale=1]{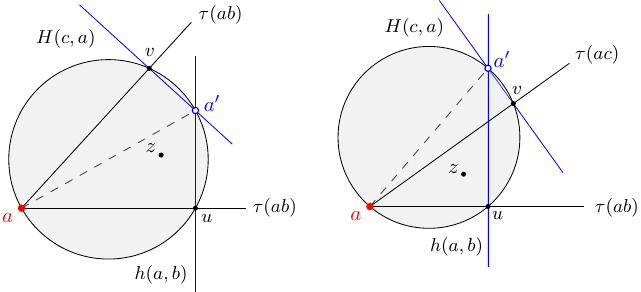}}\\
 \subfloat[]{\label{fig:prop4}\includegraphics[page=13,scale=1]{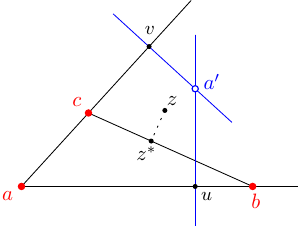}}~~~~
 \subfloat[]{\label{fig:prop5}\includegraphics[page=14,scale=1]{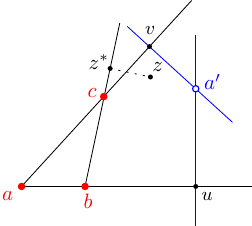}}~~~~
 \subfloat[]{\label{fig:prop6}\includegraphics[page=15,scale=1]{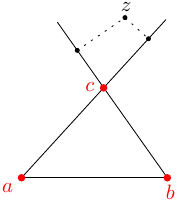}}
 \caption{\small{
		(a) Proof of Proposition~\ref{prop:2}.
		(b-c) Proof of {\bf Case 2}.
		(d) Proof of {\bf Case 3}.
	}}
 \label{fig:prop-cases}
\end{figure}

The proof is now divided into three cases.

{\bf Case 1:} The point $z$ is contained in $\Delta abc$. We then have that $z\in\D(aa')\cap\D(bb')\cap\D(cc')$ by applying Proposition~\ref{prop:2} three times.

{\bf Case 2:} The point $z$ is contained in a wedge among $\omega(abc)$, $\omega(bca)$, and $\omega(cab)$, but not in $\Delta abc$. Assume w.l.o.g.\ that $z\in\omega(abc)$. Let $z^*$ be the orthogonal projection of $z$ into $\ell(bc)$, and let $u$ and $v$ be the orthogonal projections of $a'$ into $\tau(ab)$ and $\tau(ac)$, respectively.
If $z^*\in bc$, then $z^*\in\D(bb')\cap\D(cc')$ by Proposition~\ref{prop:1}. We also have that $z^*\in\D(aa')$ by using Proposition~\ref{prop:2}. In fact, $z^*$ belongs to the convex quadrilateral of vertices $a$, $u$, $a'$, and $v$ (see Figure~\ref{fig:prop4}). Hence, we have that $z^*\in\D(aa')\cap\D(bb')\cap\D(cc')$.
Otherwise, if $z^*\notin bc$, say $z^*\in\tau(bc)\setminus bc$ w.l.o.g., then $c\in av$ which implies $c\in\D(aa')$ (see Figure~\ref{fig:prop5}). Furthermore, $c\in\D(bb')$ by Proposition~\ref{prop:1}. Hence, $c\in\D(aa')\cap\D(bb')\cap\D(cc')$.

{\bf Case 3:}  The point $z$ is not contained in any of the wedges $\omega(abc)$, $\omega(bca)$, and $\omega(cab)$. Assume w.l.o.g.\ that $z$ is such that $c\in\Delta abz$ (see Figure~\ref{fig:prop6}). The orthogonal projection of $z$ into $\ell(ac)$ belongs to $\tau(ac)\setminus ac$, so $c\in\D(aa')$ by Proposition~\ref{prop:1}. Similarly, the orthogonal projection of $z$ into $\ell(bc)$ belongs to $\tau(bc)\setminus bc$, so $c\in\D(bb')$ by Proposition~\ref{prop:1}. Hence, $c\in\D(aa')\cap\D(bb')\cap\D(cc')$.

In all of the cases, we can find a point in the intersection  $\D(aa')\cap\D(bb')\cap\D(cc')$.

\begin{theorem}[Huemer et al.~\cite{HUEMER2019}]
Let $|R|=|B|$. Any max-sum matching $\M$ of $R$ and $B$ with respect to $\|.\|^2$ is such that all the disks of $\D_\M$ have a common intersection.
\end{theorem}

{\small

\noindent{\bf Funding}: {\em Pablo P\'erez-Lantero}, Research supported by project DICYT 042332PL Vicerrector\'ia de Investigaci\'on, Desarrollo e Innovaci\'on USACH (Chile).

\noindent{\bf Data and materials availability statement}: Data and materials sharing not applicable to this article as no datasets were generated or
analyzed during the current study.

\bibliographystyle{abbrv}
\bibliography{main}

\begin{thebibliography}{10}

\bibitem{abu2023piercing}
A.~K. Abu-Affash, P.~Carmi, and M.~Maman.
\newblock Piercing diametral disks induced by edges of maximum spanning trees.
\newblock In {\em WALCOM}, pages 71--77. Springer, 2023.

\bibitem{barabanshchikova2022}
P.~Barabanshchikova and A.~Polyanskii.
\newblock Intersecting ellipses induced by a max-sum matching.
\newblock {\em Journal of Global Optimization}, pages 1--13, 2023.

\bibitem{barabanshchikova2024}
P.~Barabanshchikova and A.~Polyanskii.
\newblock Intersecting diametral balls induced by a geometric graph {II}.
\newblock {\em Discrete Mathematics}, 347(1):113694, 2024.

\bibitem{bereg2022}
S.~Bereg, O.~P. Chac{\'o}n-Rivera, D.~Flores-Pe{\~n}aloza, C.~Huemer,
  P.~P{\'e}rez-Lantero, and C.~Seara.
\newblock On maximum-sum matchings of points.
\newblock {\em Journal of Global Optimization}, 85(1):111--128, 2023.

\bibitem{andy}
D.~Eppstein.
\newblock Geometry {J}unkyard.
\newblock \url{https://www.ics.uci.edu/~eppstein/junkyard/maxmatch.html}.

\bibitem{FingerhutST97}
J.~A. Fingerhut, S.~Suri, and J.~S. Turner.
\newblock Designing least-cost nonblocking broadband networks.
\newblock {\em J. Algorithms}, 24(2):287--309, 1997.

\bibitem{helly1923}
E.~Helly.
\newblock {\"U}ber {M}engen konvexer {K}{\"o}rper mit gemeinschaftlichen
  {P}unkte.
\newblock {\em Jahresbericht der Deutschen Mathematiker-Vereinigung},
  32:175--176, 1923.

\bibitem{HUEMER2019}
C.~Huemer, P.~P\'erez-Lantero, C.~Seara, and R.~I. Silveira.
\newblock Matching points with disks with a common intersection.
\newblock {\em Discrete Mathematics}, 342(7):1885--1893, 2019.

\bibitem{perez2024center}
P.~P{\'e}rez-Lantero and C.~Seara.
\newblock Center of maximum-sum matchings of bichromatic points.
\newblock {\em Discrete Mathematics}, 347(3):113822, 2024.

\bibitem{pirahmad2022}
O.~Pirahmad, A.~Polyanskii, and A.~Vasilevskii.
\newblock Intersecting diametral balls induced by a geometric graph.
\newblock {\em Discrete \& Computational Geometry}, pages 1--18, 2022.

\bibitem{soberon2020}
P.~Sober{\'o}n and Y.~Tang.
\newblock Tverberg’s theorem, disks, and {H}amiltonian cycles.
\newblock {\em Annals of Combinatorics}, 25(4):995--1005, 2021.

\end{thebibliography}

}

\end{document}